\documentclass[11pt]{article}
\usepackage{amsmath,amssymb,amsthm,graphicx,geometry,setspace}
\usepackage[longnamesfirst]{natbib}
\newtheorem{lemma}{Lemma}
\newtheorem{theorem}{Theorem}
\newtheorem{corollary}{Corollary}
\newtheorem{proposition}{Proposition}
\newtheorem{assumption}{Assumption}
\begin{document}
\title{Welfare Analysis via Marginal Treatment Effects
}
\author{Yuya Sasaki\thanks{Y. Sasaki: Department of Economics, Vanderbilt University, VU Station B \#351819, 2301 Vanderbilt Place, Nashville, TN 37235-1819. Email: yuya.sasaki@vanderbilt.edu}\\ Department of Economics\\ Vanderbilt University
\and 
Takuya Ura\thanks{T. Ura: Department of Economics, University of California, Davis, One Shields Avenue, Davis, CA 95616. Email: takura@ucdavis.edu}\\ Department of Economics\\ University of California, Davis}
\date{} 
\maketitle 
\begin{abstract}
Consider a causal structure with endogeneity (i.e., unobserved confoundedness) in empirical data, where an instrumental variable is available. In this setting, we show that the mean social welfare function can be identified and represented via the marginal treatment effect \citep[MTE,][]{bjorklund/moffitt:1987} as the operator kernel. This representation result can be applied to a variety of statistical decision rules for treatment choice, including plug-in rules, Bayes rules, and empirical welfare maximization (EWM) rules as in \citet[Section 2.3]{hirano2020}. Focusing on the application to the EWM framework of \cite{kitagawa2018should}, we provide convergence rates of the worst case average welfare loss (regret) in the spirit of \citet{manski:2004}. 
\begin{description}
\item {\bf Keywords:} empirical welfare maximization (EWM), endogeneity, heterogeneity, marginal treatment effects (MTE), statistical decision rules, treatment choice
\item {\bf JEL Codes:} C14, C21
\end{description}
\end{abstract}

\section{Introduction}

One of the most important goals of empirical economic research is to advise policy makers on how to assign heterogeneous individuals to a treatment under consideration subject to budgetary, legal, and ethical constraints, based on evidence from empirical data.
To this goal, it is crucial to identify a social welfare function in observational data settings.
For many observational data sets used by empirical researchers, treatments are likely to be endogenously selected by rational individuals, rather than randomly assigned.
Furthermore, the effects of these treatments are often heterogeneous across individuals, even after controlling for their observable attributes.
In this light, we propose a novel method of identifying the mean social welfare function in the presence of unobserved heterogeneity in treatment effects while accounting for endogenous treatment selection in empirical data.

It is well known today that the marginal treatment effects \citep[MTE,][]{bjorklund/moffitt:1987} measure heterogeneous treatment effects, and the MTE can be identified with an instrumental variable under  endogenous treatment selection \citep[][]{heckman/vytlacil:2001,heckman/vytlacil:2005,heckman/vytlacil:2007}.
Hence, it is a natural idea to use the MTE as a building block for the identification of a social welfare function in the presence of unobserved heterogeneity and endogeneity.
In this paper, we show that the mean social welfare function can be indeed identified and represented via the MTE as the operator kernel.
Since the identification and estimation of the MTE have been well established in the existing literature \citep[e.g.,][]{heckman/vytlacil:2001,heckman/vytlacil:2005,heckman/vytlacil:2007,carneiro/lee:2009,carneiro/heckman/vytlacil:2010}, our result thus paves the way for these existing theories and methods of MTE to be directly applied to welfare analysis.

Once the mean social welfare function has been identified via the MTE, we can apply it to a variety of policy makers' statistical decision problems of treatment choice, including those based on plug-in rules, Bayes rules, and empirical welfare maximization (EWM) rules \citep[see][Section 2.3]{hirano2020}.
Focusing on the EWM rules in particular, we can take advantage of the technology developed by \cite{kitagawa2018should} to analyze properties of the EWM method in the spirit of \citet{manski:2004}.
Specifically, under both the heterogeneity and endogeneity, we can derive convergence rates of the worst case average welfare loss (regret) from the maximum empirical welfare.
As such, our result contributes to the literature by extending the scope of applicability of the EWM framework of \cite{kitagawa2018should}, that is originally based on the assumption of selection on observables or unconfoundedness, to the framework that now allows for unobserved confoundedness or endogeneity.

A recent paper by \cite{athey/wager:2020} considers an endogeneity problem in the context of the EWM, and proposes a theory based on doubly robust estimators of average treatment effects.
The result that we propose in this paper neither nests nor is nested by that of \cite{athey/wager:2020} -- these two papers play rather complementary roles.
On the one hand, \citet[Eq. (16)]{athey/wager:2020} assume homogeneous treatment effects, which implies a constant MTE, while our framework can allow for unobserved heterogeneity in treatment effects.
On the other hand, the applicability of our proposed method hinges on the identification of the MTE, while \citet{athey/wager:2020} do not need to identify the MTE for their objective.
In other words, our framework accommodates unobserved heterogeneity at the expense of assuming that the MTE is identified.
This tradeoff illustrates the complementarity between our result and the result developed by \citet{athey/wager:2020}.

Another recent paper by \cite{byambadalai:2020} also considers an endogeneity problem in the context of counterfactual welfare comparisons, and is thus closely related to this paper.
These two papers again play complementary roles.
On the one hand, \cite{byambadalai:2020} develops the partial identification and focuses on welfare gains and losses as parameters of interest, while we develop the point identification of the welfare function that is applicable to various statistical decision rules as well as welfare comparisons.
On the other hand, \cite{byambadalai:2020} imposes weak assumptions and in particular does not need to assume to identify the MTE.
This tradeoff illustrates the complementarity between our result and the result developed by \cite{byambadalai:2020}.


{\bf Relation to the Literature:}
This paper aims to contribute to the literature on statistical decisions in econometrics -- see the recent survey by \citet{hirano2020} for a comprehensive review of this subject.
In particular, we focus on an application of our representation theorem to bounding the worst case average welfare loss in the spirit of \cite{manski:2004} with the recent technology developed by \citet{kitagawa2018should}, as mentioned above.
Besides, this paper is also related to the broad literature on policy choices and welfare analysis including
\cite{manski:2004},
\cite{manski:2004,manski:2009},
\cite{dehejia2005program},
\cite{schlag2007eleven},
\cite{hirano2009asymptotics,hirano2020},
\cite{stoye2009minimax,stoye2012minimax},
\cite{chamberlain2011bayesian},
\cite{bhattacharya2012inferring}, 
 \cite{tetenov2012statistical},
\cite{armstrong2015inference},
\cite{kasy:2016},
\cite{mbakop2016model},
 \cite{kock2017optimal},
 \cite{kitagawa2018should,kitagawa2019equality},
 \cite{rai2018statistical},
\cite{sakaguchi:2019}
 \cite{viviano2019policy},
\cite{athey/wager:2020},
\cite{byambadalai:2020}, 
\cite{han2020comment}, and 
\cite{sun:2020}.
In particular, our result complements those of \cite{athey/wager:2020} and \cite{byambadalai:2020}, as mentioned above.
Also closely related is the literature on the MTE \citep{bjorklund/moffitt:1987} and its identification and estimation, including
\citet{heckman/vytlacil:2001,heckman/vytlacil:2005,heckman/vytlacil:2007},
\citet{carneiro/lee:2009},
\citet{carneiro/heckman/vytlacil:2010},
\cite{brinch/mogstad/wiswall:2017},
\citet{lee2018identifying}, and
\cite{mogstad/santos/torgovitsky:2017}
among many others.
The applicability of our representation theorem relies on the identification of the MTE from this literature.
Finally, this paper also complements the literature on policy relevant treatment effects \citep[e.g.,][]{heckman/vytlacil:2001,heckman/vytlacil:2005,heckman/vytlacil:2007,brinch/mogstad/wiswall:2017,carneiro/lokshin/umapathi:2017,mogstad/santos/torgovitsky:2017,sasaki2018estimation} which develops methods of identification, estimation and inference for average welfare gains under counterfactual policies based on the MTE.


{\bf Organization:}
The rest of this paper is organized as follows.
Section \ref{sec:model} introduces the model.
Section \ref{sec:main_result} presents the main result of representing the mean social welfare via the marginal treatment effects.
Section \ref{sec:application_to_sdr} introduces applications to three statistical decision rules.
Section \ref{sec:application_to_ewm} demonstrates the use of the representation result in the empirical welfare analysis. 
Section \ref{sec:concl} concludes.
The appendix contains mathematical proofs.

\section{The model} \label{sec:model}
Consider the model
\begingroup
\allowdisplaybreaks
\begin{align}
Y&=DY_1+(1-D)Y_0
\label{eq:y}
\\
D&=1\{\tilde{\nu}(Z)-\tilde{U}\geq 0\},
\label{eq:d}
\end{align}
\endgroup
where 
$Y$ denotes an observed outcome variable, 
$D$ denotes an observed binary treatment variable,
$Z$ denotes a vector of observed exogenous variable,
$Y_0$ and $Y_1$ denote unobserved potential outcomes under no treatment and under treatment, respectively, and
$\tilde{U}$ denotes an unobserved factor of the treatment selection.
The first equation \eqref{eq:y} models the outcome production through the potential outcome framework, and the second equation \eqref{eq:d} models the treatment selection via a threshold-crossing model.
The function $\tilde{\nu}$ in this threshold-crossing treatment assignment model \eqref{eq:d} is nonparametric and is unknown to the econometrician.

This model allows for endogeneity (unobserved confoundedness) in the sense that $(Y_0,Y_1)$ and $\tilde{U}$ may be statistically dependent even conditionally on $Z$.
For the purpose of identification, therefore, we require the vector $Z$ to consist of excluded exogenous variables (i.e., excluded instruments) as well as included exogenous variables, as formally stated in Assumption \ref{assn:MTEassumption} below.
The next assumption are standard in the recent literature on marginal treatment effects \citep[e.g.,][]{brinch/mogstad/wiswall:2017,mogstad/santos/torgovitsky:2017}.

\begin{assumption}[Model Restrictions]\label{assn:MTEassumption}
The random vector $Z$ can be written as $(Z_0',X')'$, where  
\begin{enumerate}
\item[(i)] 
$\tilde{U}$ and $Z_0$ are independent given $X$;
\item[(ii)] 
$E[Y_d\mid Z,\tilde{U}]=E[Y_d\mid X,\tilde{U}]$ and $E[Y_d^2]<\infty$; and 
\item[(iii)] 
$\tilde{U}$ is continuously distributed with a convex support conditional on $X$.
\end{enumerate}
\end{assumption}

Part (i) concerns solely about the treatment assignment model \eqref{eq:d}, and this is the only independence assumption to be imposed on the model, implying that we can allow for an arbitrary statistical dependence between the potential outcomes $(Y_0,Y_1)$ and $\tilde U$, even conditionally on $Z$.
Part (ii) states the exclusion restriction of the random sub-vector $Z_0$ of $Z$, and bounded second moments of the potential outcomes $(Y_0,Y_1)$.
Part (iii) rules out point masses and holes in the conditional distribution of $\tilde{U}$ given $X$.

For ease of analysis, by following the literature on the marginal treatment effects, we apply normalizing transformations, $U\equiv F_{\tilde{U}\mid X}(\tilde{U})$ and ${\nu}(Z)\equiv F_{\tilde{U}\mid X}(\tilde{\nu}(Z))$, in the threshold crossing model \eqref{eq:d}.
The following lemma confirms convenient properties to be used throughout the rest of the paper, as a result of these normalizing transformations under Assumption \ref{assn:MTEassumption}.

\begin{lemma}[Normalization]\label{lemma:normalization}
Suppose that Assumption \ref{assn:MTEassumption} (i) and (iii) hold. Then (i) $D=1\{\nu(Z)-U\geq 0\}$, and (ii) $U$ is distributed uniformly over $[0,1]$ conditional on $Z$. 
\end{lemma}

A proof of this lemma is provided in Appendix \ref{sec:lemma:normalization}.
Consequently, we can rewrite the threshold-crossing treatment selection model \eqref{eq:d} without loss of generality as
\begin{equation}\label{eq:d_wlog}
D=1\{\nu(Z)-U\geq 0\}
\mbox{ with } U|Z \sim \text{Uniform}(0,1).
\end{equation}
We will hereafter substitute the model model \eqref{eq:d_wlog} for the original model \eqref{eq:d} by virtue of Assumption \ref{assn:MTEassumption}.

\section{The main result}\label{sec:main_result}
In this section, we show that the social welfare function can be identified and represented via the marginal treatment effects as the operator kernel.
To this end, we first introduce and define the two key ingredients of this result, namely the social welfare function and the marginal treatment effects.

A policy maker assigns individuals with certain observed attributes $Z$ to a treatment $D=1$.
Thus, a treatment assignment rule is represented by a decision set $G \subset \mathcal{Z}$, where $\mathcal{Z}$ is a set of values that $Z$ may take.
Specifically, the decision set $G$ represents the policy in which individuals with $Z \in G$ are assigned to a treatment $D=1$ while those with $Z \not\in G$ are not.
Let $\mathcal{G}$ denote a collection all the decision sets $G$ under consideration subject to the policy makers' constraints.
With these notations, the social welfare function $W: \mathcal{G} \rightarrow \mathbb{R}$ is defined by
$$
W(G)=E[1\{Z\in G\}Y_1+1\{Z\notin G\}Y_0].
$$
Next, recall from Assumption \ref{assn:MTEassumption} that $X$ is the included sub-vector of the random vector $Z$ of exogenous variables that affects the treatment assignment, and also recall from Lemma \ref{lemma:normalization} or Equation \eqref{eq:d_wlog} that $U$ is the normalized unobserved factor of the treatment selection.
The marginal treatment effect \citep[MTE,][]{bjorklund/moffitt:1987} is defined by 
$$
MTE(u,x)=E[Y_1-Y_0\mid U=u,X=x]. 
$$
With these definitions of the social welfare function and the marginal treatment effect, we now state the following theorem as the main result of this paper.

\begin{theorem}[Representation]\label{theorem:MTErepresentationEWM}
Under Assumption \ref{assn:MTEassumption}, one has
\begin{equation}\label{eq:key_identification} 
W(G)=E[Y_0]+E\left[1\{Z\in G\}\int_0^1MTE(u,X)du\right]
\qquad\text{for every $G \in \mathcal{G}$.}
\end{equation}
\end{theorem}
\begin{proof}
Since $Y_1$ and $Y_0$ are integrable under Assumption \ref{assn:MTEassumption} (ii), we have
\begingroup
\allowdisplaybreaks
\begin{align*}
W(G)
=&
E[1\{Z\in G\}Y_1+1\{Z\notin G\}Y_0]\\
=&
E[1\{Z\in G\}(Y_1-Y_0)]+E[Y_0]. 
\end{align*}
\endgroup
Now, the statement of this theorem follows from 
\begingroup
\allowdisplaybreaks
\begin{align*}
E[1\{Z\in G\}(Y_1-Y_0)]
=&
E[1\{Z\in G\}E[Y_1-Y_0\mid Z,U]]\\
=&
E[1\{Z\in G\}MTE(U,X)]\\
=&
E[1\{Z\in G\}E[MTE(U,X)\mid Z]]\\
=&
E\left[1\{Z\in G\}\int_0^1MTE(u,X)du\right],
\end{align*}
\endgroup
where the first equality follows from the law of iterated expectations, the second equality follows from Assumption \ref{assn:MTEassumption} (ii) and the definition of $MTE$, the third equality follows from another application of the law of iterated expectations, and the fourth equality follows from Lemma \ref{lemma:normalization} under Assumption \ref{assn:MTEassumption} (i) and (iii). 
\end{proof}

The representation \eqref{eq:key_identification} of the mean social welfare via $MTE(u,x)$ is the key result of this paper. 
Because there is an existing literature on identification and estimation for the MTE \citep[e.g.,][]{heckman/vytlacil:2001,heckman/vytlacil:2005,heckman/vytlacil:2007,carneiro/lee:2009,carneiro/heckman/vytlacil:2010}, our result \eqref{eq:key_identification} paves the way for empirical welfare analysis under the potential endogeneity or unobserved confoundedness based on the existing identification and estimation methods of the MTE.  
We present a few examples of such applications in Sections \ref{sec:application_to_sdr} and \ref{sec:application_to_ewm}.

Finally, we remark on relations to and differences from \cite{kitagawa2018should}, who use the representation 
$$
W(G) = E[Y_0] + E\left[1\{Z \in G\} \tau(X)\right],
$$
where $\tau(x) \equiv E[Y_1-Y_0|X=x]$.
Our representation \eqref{eq:key_identification} is closely related to this representation.
Under the unconfoundedness assumption, \cite{kitagawa2018should} use the identification of $\tau(x)$ by $E[Y|D=1,X=x]-E[Y|D=0,X=x]$.
On the other hand, under the unobserved confoundedness in our setup, the corresponding operator kernel $\tau(x)$ is not identified by $E[Y|D=1,X=x]-E[Y|D=0,X=x]$ in general.
Instead, we propose to take advantage of the identification and estimation of $MTE$ from the literature on the marginal treatment effects.

\section{Applications to statistical decision rules}\label{sec:application_to_sdr}

Once we obtain the representation \eqref{eq:key_identification} of the mean social welfare via $MTE(u,x)$, we may apply it to a variety of policy makers' statistical decision problems for treatment choice. 
In this section, following \citet[Section 2.3]{hirano2020}, we introduce applications to the three popular statistical decision rules:
1. plug-in rules,
2. Bayes rules, and
3. empirical welfare maximization rules.
In Section \ref{sec:application_to_ewm}, we discuss the empirical welfare maximization rules in further details based on recent technologies. 

\subsection*{Plug-in rules}
Suppose that the distribution of $Z$ is parametrized by $\delta$ and we have an estimator $(\hat\delta,\widehat{MTE})$ for $(\delta,MTE)$. 
In this case, we can estimate the maximizer for the population social welfare by maximizing 
$$
\int 1\{z\in G\}\int_0^1\widehat{MTE}(u,x)dud\mu_{\hat\delta}(z)
$$
over $G\in\mathcal{G}$, where $\mu_{\delta}(\cdot)$ is the probability measure of $Z$ indexed by $\delta$ and we may ignore the term $E[Y_0]$ in \eqref{eq:key_identification} since it does not affect the maximization problem over $G\in\mathcal{G}$. 
Note that $(\hat\delta,\widehat{MTE})$ is an estimator, and so a maximizer of the above objective function is a statistical decision rule, i.e., it is a function of the observed data. 

\subsection*{Bayes rules}
Suppose that the distribution of $Z$ is parametrized by $\delta$ and $MTE$ is parametrized by $\eta$. 
In addition, suppose that we have a prior probability measure, denoted by $\pi_{\text{prior}}$, of $(\delta,\eta)$, and we can construct the posterior distribution, denoted by $\pi_{\text{posterior}}$, of $(\delta,\eta)$ by Bayesian updating. 
Given the posterior probability measure of $(\delta,\eta)$ and our representation \eqref{eq:key_identification}, we can construct the Bayes welfare
$$
\int \left(\int 1\{z\in G\}\int_0^1{MTE}_\eta(u,x)dud\mu_\delta(z)\right)d\pi_{\text{posterior}}(\delta,\eta).
$$
The Bayes rule is the maximizer of the Bayes welfare over $G \in \mathcal{G}$.  

\subsection*{Empirical welfare maximization rules}
The empirical welfare maximization rule uses the empirical distribution of $Z$ and an estimator $\widehat{MTE}$ for $MTE$. 
Namely, with a random sample $\{Z_1,\ldots,Z_n\}$ of size $n$, we can define the empirical welfare by 
$$
E_n\left[1\{Z\in G\}\int_0^1\widehat{MTE}(u,X)du\right],
$$
where $E_n$ denotes the sample average operator, i.e., $E_n f(Z) = n^{-1} \sum_{i=1}^n f(Z_i)$ for any measurable function $f$.
The empirical welfare  maximization rule selects the maximizer of this empirical welfare over $G \in \mathcal{G}$.
The following section presents more detailed analyses of the asymptotic properties of the maximum of this empirical welfare relative the population mean welfare under the oracle action.

\section{Applications to empirical welfare maximization}\label{sec:application_to_ewm}

We demonstrate applications of the representation result \eqref{eq:key_identification} to empirical welfare maximization in this section.
For the purpose of exposition of the core idea, we first focus on the case where the mapping $(u,x)\mapsto MTE(u,x)$ is known by a researcher in Section \ref{sec:known}.
We then present the case where the mapping $(u,x)\mapsto MTE(u,x)$ is unknown by a researcher and thus needs to be estimated in Section \ref{sec:unknown}.

\subsection{Empirical welfare maximization with known MTE}\label{sec:known}
In this section, we assume that we know the mapping $(u,x)\mapsto MTE(u,x)$. 
The empirical welfare maximizer in this setting is given by
$$
\hat{G}_{EWM}\in\arg\max_{G\in\mathcal{G}}E_n\left[1\{Z\in G\}\int_0^1MTE(u,X)du\right].
$$ 
We present a uniform asymptotic analysis of the maximum empirical welfare $W(\hat{G}_{EWM})$, relative to the population mean social welfare under the oracle action, denoted by
$$
W_{\mathcal{G}}=\sup_{G\in\mathcal{G}}W(G).
$$
To this end, consider the following assumption.

\begin{assumption}\label{assn:kitagawatetenov2.1}
(i) $|\int_0^1MTE(u,X)du|\leq\bar{M} < \infty$ a.s. for the class  $\mathcal{P}(\bar{M})$ of distributions of $(Y_0,Y_1,D,Z)$.
(ii) $\mathcal{G}$ has a finite VC-dimension $v<\infty$ and is countable. 
\end{assumption}

The above assumption is a modification of Assumption 2.1 (BO)-(VC)  in \cite{kitagawa2018should} tailored to our framework with the marginal treatment effects. 
Assumption \ref{assn:kitagawatetenov2.1} (i) requires a bounded integral of the marginal treatment effect function. 
As a sufficient condition, it holds when the outcome variable is bounded by some constant which is naturally satisfied in some applications. 
Assumption \ref{assn:kitagawatetenov2.1} (ii) restricts the complexity of the class of treatment functions $Z\mapsto 1\{Z\in G\}$. 
As a sufficient condition, when $X$ has a finite support, this assumption will automatically hold where $v$ is the cardinality of the power set for the support of $X$.  
\citet[p.598]{kitagawa2018should} collects a few examples of $\mathcal{G}$ with finite VC-dimensions.  

The following corollary provides a convergence rate of the worst case average welfare loss (regret) by the empirical welfare maximization.

\begin{corollary}\label{theorem:known}
Under Assumptions \ref{assn:MTEassumption} and \ref{assn:kitagawatetenov2.1}, one has
$$
\sup_{P\in\mathcal{P}(\bar{M})}E_{P^n}\left[W_{\mathcal{G}}-W(\hat{G}_{EWM})\right]\leq 2C_1\bar{M}\sqrt{\frac{v}{n}},
$$
where $C_1$ is a universal constant. 
\end{corollary}

A proof is provided in Appendix \ref{sec:theorem:known}.
Corollary \ref{theorem:known} implies that no treatment assignment rule based on empirical data will achieve a minimax rate that is faster than $n^{-1/2}$, and shows that $\hat{G}_{EWM}$ is the minimax rate optimal over $\mathcal{P}(\bar{M})$.
This corollary extends and is a counterpart of Theorem 2.1 of \cite{kitagawa2018should}. 

\subsection{Empirical welfare maximization with unknown MTE}\label{sec:unknown}
In this section, we consider the case where the mapping $(u,x)\mapsto MTE(u,x)$ is unknown by a researcher and thus needs to be estimated from empirical data. 
The empirical welfare maximizer in this setting is given by
$$
\hat{G}_{hybrid}\in\arg\max_{G\in\mathcal{G}}E_n\left[1\{Z\in G\}\int_0^1\widehat{MTE}(u,X)du\right],
$$ 
where $\widehat{MTE}$ is an estimator of $MTE$. 
The existing literature on marginal treatment effects provides a list of alternative estimators $\widehat{MTE}$ of $MTE$.
We therefore first provide a general sufficient condition in terms of $\widehat{MTE}$ that accommodate a wide range of possible estimators in Section \ref{sec:sufficient_conditions}.
This will be followed up by a specific estimator $\widehat{MTE}$ with lower level primitive conditions tailored to it in Section \ref{sec:parametric_estimation}.

\subsubsection{A sufficient condition}\label{sec:sufficient_conditions}

A general high-level assumption for an estimator $\widehat{MTE}$ of $MTE$ is that it entails a uniform convergence rate $\psi_n^{-1}$ in the mean absolute value of $\int_0^1 \widehat{MTE}(u,X)du$ for $\int_0^1 MTE(u,X)du$, as formally stated below.

\begin{assumption}\label{assn:estimation_error}
For a class of data generating processes $\mathcal{P}_m$, there exists a sequence $\psi_n\rightarrow\infty$ such that 
$$
\limsup_{n\rightarrow\infty}\sup_{P\in\mathcal{P}_m}\psi_nE_{P^n}\left[E_n\left[\left|\int_0^1(\widehat{MTE}(u,X)-{MTE}(u,X))du\right|\right]\right]<\infty.
$$
\end{assumption}

This sufficient condition leads to the rate, $\psi_n^{-1} \vee n^{-1/2}$, of convergence for the worst case average welfare loss (regret), as formally stated as a corollary below.

\begin{corollary}\label{corollary:sufficient}
Under Assumptions \ref{assn:MTEassumption}, \ref{assn:kitagawatetenov2.1}, and \ref{assn:estimation_error},
$$
\sup_{P\in\mathcal{P}_m\cap\mathcal{P}(\bar{M})}E_{P^n}\left[W_{\mathcal{G}}-W(\hat{G}_{hybrid})\right]=O(\psi_n^{-1}\vee n^{-1/2}).
$$
\end{corollary}

A proof of this corollary is provided in Appendix \ref{sec:corollary:sufficient}.
It serves as a counterpart of Theorem 2.5 of \cite{kitagawa2018should}.
Different estimators $\widehat{MTE}$ of $MTE$ in general entail different convergence rates $\psi_n^{-1}$.
We next present a concrete estimator with lower-level sufficient conditions for the high-level condition in Assumption \ref{assn:estimation_error}.

\subsubsection{Parametric estimation for $MTE(u,x)$}\label{sec:parametric_estimation}
By \cite{heckman/vytlacil:1999,heckman/vytlacil:2001,heckman/vytlacil:2005}, the MTE can be identified from data via the equality
$$
MTE(u,x)=\frac{E[Y\mid \nu(Z)=u,X=x]}{\partial u}\equiv LIV(u,x).
$$
In this section, we consider the parametric regression function for $E[Y\mid \nu(Z)=u,X=x]$ by 
$$
E[Y\mid \nu(Z)=u,X=x]=x'\beta_0+x'(\beta_1-\beta_0)u+\sum_{k=2}^K\alpha_ku^k
$$
following \citet[][Section 4.3]{cornelissen/dustmann/raute/schonberg:2016} on the survey of the marginal treatment effects for labor economists.
Let $\hat\nu(Z)$ denote some estimator of the propensity score $\nu(Z)$, and define 
\begingroup
\allowdisplaybreaks
\begin{align*}
{\mathcal{X}}&=((1-\nu(Z))X',\nu(Z)X',\nu(Z)^2,\ldots,\nu(Z)^K)',
\\
\hat{\mathcal{X}}&=((1-\hat\nu(Z))X',\hat\nu(Z)X',\hat\nu(Z)^2,\ldots,\hat\nu(Z)^K)',
\qquad\text{and}\\
\theta&=(\beta_0,\beta_1,\alpha_2,\ldots,\alpha_K).
\end{align*}
\endgroup

Let $\hat\theta=(\hat\beta_0,\hat\beta_1,\hat\alpha_2,\ldots,\hat\alpha_K)$ be the OLS estimator for $\theta$ by regressing $Y$ on $\hat{\mathcal{X}}$, that is, 
$$
\hat\theta=E_n\left[\hat{\mathcal{X}}\hat{\mathcal{X}}'\right]^{-1}E_n\left[\hat{\mathcal{X}}Y\right].
$$
Then, $MTE$ can be simply estimated by the following linear functional of $\hat\theta$.
$$
\widehat{MTE}(u,x)=x'(\hat\beta_1-\hat\beta_0)+\sum_{k=2}^Kk\hat\alpha_ku^{k-1}.
$$
Therefore, the operator kernel of in our representation \eqref{eq:key_identification} can be estimated by the simple linear expression
$$
\int_0^1\widehat{MTE}(u,x)du=x'(\hat\beta_1-\hat\beta_0)+\sum_{k=2}^K\hat\alpha_k.
$$
For this concrete estimator, we provide a set of lower-level conditions in the proposition below that guarantee the aforementioned high-level condition in Assumption \ref{assn:estimation_error} to be satisfied.

\begin{proposition}\label{theorem:estimation_error}
Let $C$ and $c$ be positive constants, and $\psi_n$ be a sequence with $\psi_n\geq n^{1/2}$. 
Suppose that the parameter space for $\theta$ is compact so that for sufficiently large $n$
\begin{equation}\label{eq:para_compact}
\|\hat\theta\|+\|\theta\|\leq C\mbox{ almost surely.} 
\end{equation}
Furthermore, suppose that $\mathcal{P}_m$ is a class of data generating processes such that 
\begingroup
\allowdisplaybreaks
\begin{align}\label{eq:P_hat_estimation}
\limsup_{n\rightarrow\infty}\sup_{P\in\mathcal{P}_m}\psi_nE_{P^n}\left[\max_{i=1,\ldots,n}|\hat\nu(Z_i)-\nu(Z_i)|^2\right]^{1/2}<\infty,&
\\
\label{eq:X_compact}
\max\{E\left[\|X\|^4\right],E\left[|Y|^4\right]\}<C,& \qquad\text{and}
\\
\label{eq:mini_eigen}
\lambda_{\min}\left(E\left[\mathcal{X}\mathcal{X}'\right]\right)\geq c.&
\end{align}
\endgroup
Then, Assumption \ref{assn:estimation_error} is satisfied.
\end{proposition}

A proof is provided in Appendix \ref{sec:theorem:estimation_error}.
Except for \eqref{eq:P_hat_estimation}, all the conditions stated in Proposition \ref{theorem:estimation_error} can be considered as regularity conditions.  
The condition in \eqref{eq:P_hat_estimation} requires a convergence rate for an estimator $\hat\nu(z)$ of the propensity score $\nu(z)$ uniformly over the data generating processes. 
This condition can be checked with specific propensity score estimators.
For example, we can use the local polynomial estimator $\hat\nu(z)$ for $\nu(z)$, for which \citet[Appendix H]{kitagawa2018should} derive a uniform convergence rate. 
Specifically, the convergence in \eqref{eq:P_hat_estimation} follows directly from their Lemma E.4 (ii).
For another example, one could consider a linear propensity score model $\hat\nu(z) = p(z)' \hat\gamma$ and its least squares estimator $\hat\nu(z) = p(z)' \hat\gamma$.
In this case, \eqref{eq:P_hat_estimation} can be satisfied with $\psi_n = \sqrt{n}$, so that the convergence rate for the worst case average welfare loss (regret) in \eqref{corollary:sufficient} holds with the parametric root $n$ rate.


\section{Conclusion}\label{sec:concl}

An important research goal for empirical economists is to provide policy makers with guidance on how heterogeneous individuals can be assigned to a treatment under consideration based on evidence from empirical data. 
To this goal, it is essential to identify a social welfare function from observational data.
For many observational data sets used in empirical economic research, treatments are likely to be endogenously selected by rational agents.
Furthermore, the effects of these treatments are often heterogeneous even after controlling for observed attributes.
In this light, given the abilities of the marginal treatment effects to measure heterogeneous treatment effects, we propose the usage of the marginal treatment effects for identifying the mean social welfare function in the presence of unobserved heterogeneity in treatment effects while accounting for endogenous treatment selection in the empirical data.
Our main result, Theorem \ref{theorem:MTErepresentationEWM}, establishes that the mean social welfare can be represented via the marginal treatment effects as the operator kernel. 
We introduce applications of this main result to a few of policy makers' statistical decision problems, such as the plug-in rules, the Bayes rules, and in particular the empirical welfare maximization rule.
Focusing on the last application, we derive convergence rates of the worst case average welfare loss (regret) from the maximum empirical welfare under alternative scenarios.
The proposed representation in Theorem \ref{theorem:MTErepresentationEWM} can be beneficial as it allows the machinery developed in the existing literature on the marginal treatment effects to be directly applicable to the variety of empirical welfare analysis.

\bigskip

{\singlespacing
\bibliography{mybib}

\begin{thebibliography}{37}
\newcommand{\enquote}[1]{``#1''}
\expandafter\ifx\csname natexlab\endcsname\relax\def\natexlab#1{#1}\fi

\bibitem[\protect\citeauthoryear{Armstrong and Shen}{Armstrong and
  Shen}{2015}]{armstrong2015inference}
\textsc{Armstrong, T. and S.~Shen} (2015): \enquote{Inference on Optimal
  Treatment Assignments,} Cowles Foundation Discussion Paper.

\bibitem[\protect\citeauthoryear{Athey and Wager}{Athey and
  Wager}{2020}]{athey/wager:2020}
\textsc{Athey, S. and S.~Wager} (2020): \enquote{Policy Learning with
  Observational Data,} \emph{Econometrica}, forthcoming.

\bibitem[\protect\citeauthoryear{Bhattacharya and Dupas}{Bhattacharya and
  Dupas}{2012}]{bhattacharya2012inferring}
\textsc{Bhattacharya, D. and P.~Dupas} (2012): \enquote{Inferring Welfare
  Maximizing Treatment Assignment under Budget Constraints,} \emph{Journal of
  Econometrics}, 167, 168--196.

\bibitem[\protect\citeauthoryear{Bj\"orklund and Moffitt}{Bj\"orklund and
  Moffitt}{1987}]{bjorklund/moffitt:1987}
\textsc{Bj\"orklund, A. and R.~Moffitt} (1987): \enquote{The Estimation of Wage
  Gains and Welfare Gains in Self-Selection Models,} \emph{The Review of
  Economics and Statistics}, 69, 42--49.

\bibitem[\protect\citeauthoryear{Brinch, Mogstad, and Wiswall}{Brinch
  et~al.}{2017}]{brinch/mogstad/wiswall:2017}
\textsc{Brinch, C.~N., M.~Mogstad, and M.~Wiswall} (2017): \enquote{Beyond LATE
  with a Discrete Instrument,} \emph{Journal of Political Economy}, 125,
  985--1039.

\bibitem[\protect\citeauthoryear{Byambadalai}{Byambadalai}{2020}]{byambadalai:2020}
\textsc{Byambadalai, U.} (2020): \enquote{Identification and Inference for
  Welfare Gains without Unconfoundedness,} Working paper.

\bibitem[\protect\citeauthoryear{Carneiro, Heckman, and Vytlacil}{Carneiro
  et~al.}{2010}]{carneiro/heckman/vytlacil:2010}
\textsc{Carneiro, P., J.~J. Heckman, and E.~Vytlacil} (2010):
  \enquote{Evaluating Marginal Policy Changes and the Average Effect of
  Treatment for Individuals at the Margin,} \emph{Econometrica}, 78, 377--394.

\bibitem[\protect\citeauthoryear{Carneiro and Lee}{Carneiro and
  Lee}{2009}]{carneiro/lee:2009}
\textsc{Carneiro, P. and S.~Lee} (2009): \enquote{Estimating Distributions of
  Potential Outcomes Using Local Instrumental Variables with an Application to
  Changes in College Enrollment and Wage Inequality,} \emph{Journal of
  Econometrics}, 149, 191--208.

\bibitem[\protect\citeauthoryear{Carneiro, Lokshin, and Umapathi}{Carneiro
  et~al.}{2017}]{carneiro/lokshin/umapathi:2017}
\textsc{Carneiro, P., M.~Lokshin, and N.~Umapathi} (2017): \enquote{Average and
  Marginal Returns to Upper Secondary Schooling in Indonesia,} \emph{Journal of
  Applied Econometrics}, 32, 16--36.

\bibitem[\protect\citeauthoryear{Chamberlain}{Chamberlain}{2011}]{chamberlain2011bayesian}
\textsc{Chamberlain, G.} (2011): \enquote{Bayesian Aspects of Treatment
  Choice,} \emph{The Oxford Handbook of Bayesian Econometrics}, 11--39.

\bibitem[\protect\citeauthoryear{Cornelissen, Dustmann, Raute, and
  Sch\"onberg}{Cornelissen
  et~al.}{2016}]{cornelissen/dustmann/raute/schonberg:2016}
\textsc{Cornelissen, T., C.~Dustmann, A.~Raute, and U.~Sch\"onberg} (2016):
  \enquote{From LATE to MTE: Alternative Methods for the Evaluation of Policy
  Interventions,} \emph{Labour Economics}, 41, 47--60.

\bibitem[\protect\citeauthoryear{Dehejia}{Dehejia}{2005}]{dehejia2005program}
\textsc{Dehejia, R.~H.} (2005): \enquote{Program Evaluation as a Decision
  Problem,} \emph{Journal of Econometrics}, 125, 141--173.

\bibitem[\protect\citeauthoryear{Han}{Han}{2020}]{han2020comment}
\textsc{Han, S.} (2020): \enquote{Comment: Individualized Treatment Rules Under
  Endogeneity,} \emph{Journal of the American Statistical Association}.

\bibitem[\protect\citeauthoryear{Heckman and Vytlacil}{Heckman and
  Vytlacil}{2001}]{heckman/vytlacil:2001}
\textsc{Heckman, J.~J. and E.~Vytlacil} (2001): \enquote{Policy-Relevant
  Treatment Effects,} \emph{American Economic Review}, 91, 107--111.

\bibitem[\protect\citeauthoryear{Heckman and Vytlacil}{Heckman and
  Vytlacil}{2005}]{heckman/vytlacil:2005}
---\hspace{-.1pt}---\hspace{-.1pt}--- (2005): \enquote{Structural Equations,
  Treatment Effects, and Econometric Policy Evaluation,} \emph{Econometrica},
  73, 669--738.

\bibitem[\protect\citeauthoryear{Heckman and Vytlacil}{Heckman and
  Vytlacil}{2007}]{heckman/vytlacil:2007}
---\hspace{-.1pt}---\hspace{-.1pt}--- (2007): \enquote{Econometric Evaluation
  of Social Programs, Part II: Using the Marginal Treatment Effect to Organize
  Alternative Econometric Estimators to Evaluate Social Programs, and to
  Forecast their Effects in New Environments,} in \emph{Handbook of
  Econometrics}, ed. by J.~J. Heckman and E.~E. Leamer, Elsevier, vol.~6,
  chap.~71, 4875--5143.

\bibitem[\protect\citeauthoryear{Heckman and Vytlacil}{Heckman and
  Vytlacil}{1999}]{heckman/vytlacil:1999}
\textsc{Heckman, J.~J. and E.~J. Vytlacil} (1999): \enquote{Local instrumental
  variables and latent variable models for identifying and bounding treatment
  effects,} \emph{Proceedings of the national Academy of Sciences}, 96,
  4730--4734.

\bibitem[\protect\citeauthoryear{Hirano and Porter}{Hirano and
  Porter}{2009}]{hirano2009asymptotics}
\textsc{Hirano, K. and J.~R. Porter} (2009): \enquote{Asymptotics for
  Statistical Treatment Rules,} \emph{Econometrica}, 77, 1683--1701.

\bibitem[\protect\citeauthoryear{Hirano and Porter}{Hirano and
  Porter}{2020}]{hirano2020}
---\hspace{-.1pt}---\hspace{-.1pt}--- (2020): \enquote{Asymptotic analysis of
  statistical decision rules in econometrics,} in \emph{Handbook of
  Econometrics}, Elsevier, vol.~7.

\bibitem[\protect\citeauthoryear{Kasy}{Kasy}{2016}]{kasy:2016}
\textsc{Kasy, M.} (2016): \enquote{Partial identification, distributional
  preferences, and the welfare ranking of policies.} \emph{Review of Economics
  and Statistics}, 98, 111--131.

\bibitem[\protect\citeauthoryear{Kitagawa and Tetenov}{Kitagawa and
  Tetenov}{2018}]{kitagawa2018should}
\textsc{Kitagawa, T. and A.~Tetenov} (2018): \enquote{Who should be treated?
  empirical welfare maximization methods for treatment choice,}
  \emph{Econometrica}, 86, 591--616.

\bibitem[\protect\citeauthoryear{Kitagawa and Tetenov}{Kitagawa and
  Tetenov}{2019}]{kitagawa2019equality}
---\hspace{-.1pt}---\hspace{-.1pt}--- (2019): \enquote{Equality-Minded
  Treatment Choice,} \emph{Journal of Business \& Economic Statistics}, 1--14.

\bibitem[\protect\citeauthoryear{Kock and Thyrsgaard}{Kock and
  Thyrsgaard}{2017}]{kock2017optimal}
\textsc{Kock, A.~B. and M.~Thyrsgaard} (2017): \enquote{Optimal Sequential
  Treatment Allocation,} \emph{arXiv preprint arXiv:1705.09952}.

\bibitem[\protect\citeauthoryear{Lee and Salani{\'e}}{Lee and
  Salani{\'e}}{2018}]{lee2018identifying}
\textsc{Lee, S. and B.~Salani{\'e}} (2018): \enquote{Identifying Effects of
  Multivalued Treatments,} \emph{Econometrica}, 86, 1939--1963.

\bibitem[\protect\citeauthoryear{Manski}{Manski}{2004}]{manski:2004}
\textsc{Manski, C.~F.} (2004): \enquote{Statistical Treatment Rules for
  Heterogeneous Populations,} \emph{Econometrica}, 72, 1221--1246.

\bibitem[\protect\citeauthoryear{Manski}{Manski}{2009}]{manski:2009}
---\hspace{-.1pt}---\hspace{-.1pt}--- (2009): \emph{Identification for
  Prediction and Decision.}, Harvard University Press.

\bibitem[\protect\citeauthoryear{Mbakop and Tabord-Meehan}{Mbakop and
  Tabord-Meehan}{2016}]{mbakop2016model}
\textsc{Mbakop, E. and M.~Tabord-Meehan} (2016): \enquote{Model Selection for
  Treatment Choice: Penalized Welfare Maximization,} \emph{arXiv preprint
  arXiv:1609.03167}.

\bibitem[\protect\citeauthoryear{Mogstad, Santos, and Torgovitsky}{Mogstad
  et~al.}{2018}]{mogstad/santos/torgovitsky:2017}
\textsc{Mogstad, M., A.~Santos, and A.~Torgovitsky} (2018): \enquote{Using
  Instrumental Variables for Inference About Policy Relevant Treatment
  Parameters,} \emph{Econometrica}, 86, 1589--1619.

\bibitem[\protect\citeauthoryear{Rai}{Rai}{2018}]{rai2018statistical}
\textsc{Rai, Y.} (2018): \enquote{Statistical Inference for Treatment
  Assignment Policies,} Working Paper.

\bibitem[\protect\citeauthoryear{Sakaguchi}{Sakaguchi}{2019}]{sakaguchi:2019}
\textsc{Sakaguchi, S.} (2019): \enquote{Estimating Optimal Dynamic Treatment
  Assignment Rules under Intertemporal Budget Constraints,} Working paper.

\bibitem[\protect\citeauthoryear{Sasaki and Ura}{Sasaki and
  Ura}{2018}]{sasaki2018estimation}
\textsc{Sasaki, Y. and T.~Ura} (2018): \enquote{Estimation and Inference for
  Policy Relevant Treatment Effects,} \emph{arXiv preprint arXiv:1805.11503}.

\bibitem[\protect\citeauthoryear{Schlag}{Schlag}{2007}]{schlag2007eleven}
\textsc{Schlag, K.~H.} (2007): \enquote{Eleven\% Designing Randomized
  Experiments under Minimax Regret,} Unpublished manuscript, European
  University Institute, Florence.

\bibitem[\protect\citeauthoryear{Stoye}{Stoye}{2009}]{stoye2009minimax}
\textsc{Stoye, J.} (2009): \enquote{Minimax Regret Treatment Choice with Finite
  Samples,} \emph{Journal of Econometrics}, 151, 70--81.

\bibitem[\protect\citeauthoryear{Stoye}{Stoye}{2012}]{stoye2012minimax}
---\hspace{-.1pt}---\hspace{-.1pt}--- (2012): \enquote{Minimax Regret Treatment
  Choice with Covariates or with Limited Validity of Experiments,}
  \emph{Journal of Econometrics}, 166, 138--156.

\bibitem[\protect\citeauthoryear{Sun}{Sun}{2020}]{sun:2020}
\textsc{Sun, L.} (2020): \enquote{Empirical Welfare Maximization with
  Constraints,} Working paper.

\bibitem[\protect\citeauthoryear{Tetenov}{Tetenov}{2012}]{tetenov2012statistical}
\textsc{Tetenov, A.} (2012): \enquote{Statistical Treatment Choice Based on
  Asymmetric Minimax Regret Criteria,} \emph{Journal of Econometrics}, 166,
  157--165.

\bibitem[\protect\citeauthoryear{Viviano}{Viviano}{2019}]{viviano2019policy}
\textsc{Viviano, D.} (2019): \enquote{Policy Targeting under Network
  Interference,} \emph{arXiv preprint arXiv:1906.10258}.

\end{thebibliography}
}

\appendix
\section{Proofs}
\subsection{Proof of Lemma \ref{lemma:normalization}}\label{sec:lemma:normalization}
\begin{proof}
The first statement of this lemma follows because Assumption \ref{assn:MTEassumption} (iii) implies that $F_{\tilde{U}\mid X=x}$ is strictly increasing and therefore 
$$
D=1\{\tilde{\nu}(Z)-\tilde{U}\geq 0\}=1\{F_{\tilde{U}\mid X}(\tilde{\nu}(Z))-F_{\tilde{U}\mid X}(\tilde{U})\geq 0\}=1\{\nu(Z)-U\geq 0\}.
$$ 
The second statement follows because Assumption \ref{assn:MTEassumption} (i) and (iii) implies 
$$
P(U\leq u\mid Z)=P(F_{\tilde{U}\mid X}(\tilde{U})\leq u\mid Z)=P(\tilde{U}\leq F_{\tilde{U}\mid X}^{-1}(u)\mid Z)=P(\tilde{U}\leq F_{\tilde{U}\mid X}^{-1}(u)\mid X)=u.
$$
These complete a proof of the lemma.
\end{proof}

\subsection{Proof of Corollary \ref{theorem:known}}\label{sec:theorem:known}
\begin{proof}
Define the function $f$ by
$$
f(Z;G)=1\{Z\in G\}\int_0^1MTE(u,X)du.
$$
Let $\mathcal{F}=\{f(\ \cdot\ ;G): G\in\mathcal{G}\}$.
Then, $\mathcal{F}$ is a class of uniformly bounded functions with $\|f\|_{\infty}\leq\bar{M}$ for all $f\in\mathcal{F}$. 
From Assumption \ref{assn:kitagawatetenov2.1}, it follows that $\mathcal{F}$ is of a VC-subgraph class with VC-dimension $v<\infty$. 
By \citet[Lemma A.4]{kitagawa2018should},  
$$
E_{P^n}\left[\sup_{f\in\mathcal{F}}\left|E_n[f]-E_P[f]\right|\right]\leq C_1\bar{M}\sqrt{\frac{v}{n}},
$$
where $C_1$ is a universal constant defined in \citet[Lemma A.4]{kitagawa2018should}. 
Now, define 
\begingroup
\allowdisplaybreaks
\begin{align*}
\bar{W}(G)&=E\left[1\{Z\in G\}\int_0^1{MTE}(u,X)du\right]
\qquad\text{and}\\
\bar{W}_n(G)&=E_n\left[1\{Z\in G\}\int_0^1{MTE}(u,X)du\right].
\end{align*}
\endgroup
Then, we have 
\begin{align}
\sup_{P\in\mathcal{P}(\bar{M})}E_{P^n}\left[\sup_{{G}\in\mathcal{G}}|\bar{W}_n(G)-\bar{W}(G)|\right]
&=
\sup_{P\in\mathcal{P}(\bar{M})}E_{P^n}\left[\sup_{f\in\mathcal{F}}\left|E_n[f]-E_P[f]\right|\right]
\notag\\
&\leq 
C_1\bar{M}\sqrt{\frac{v}{n}}.
\label{eq:lemmaA4}
\end{align}
Following the derivations in \citet[Eq. (2.2)]{kitagawa2018should},
we have for any $\tilde{G}\in\mathcal{G}$ that
\begingroup
\allowdisplaybreaks
\begin{align}
W(\tilde{G})-W(\hat{G}_{EWM})
=&
\bar{W}(\tilde{G})-\bar{W}_n(\hat{G}_{EWM})+\bar{W}_n(\hat{G}_{EWM})-\bar{W}(\hat{G}_{EWM})
\notag\\
\leq&
\bar{W}(\tilde{G})-\bar{W}_n(\tilde{G})+\sup_{{G}\in\mathcal{G}}|\bar{W}_n(G)-\bar{W}(G)|
\notag\\
\leq&
2\sup_{{G}\in\mathcal{G}}|\bar{W}_n(G)-\bar{W}(G)|,
\label{eq:eq22}
\end{align}
\endgroup
where the first inequality uses $\bar{W}_n(\hat{G}_{EWM})\geq \bar{W}_n(\tilde{G})$. 
$$
\sup_{P\in\mathcal{P}(\bar{M})}E_{P^n}\left[W_{\mathcal{G}}-W(\hat{G}_{EWM})\right]\leq 2C_1\bar{M}\sqrt{\frac{v}{n}}
$$
follows from \eqref{eq:lemmaA4} and \eqref{eq:eq22}.
\end{proof}

\subsection{Proof of Corollary \ref{corollary:sufficient}}\label{sec:corollary:sufficient}
\begin{proof}
Define 
$$
\hat{\bar{W}}_n(G)=E_n\left[1\{Z\in G\}\int_0^1\widehat{MTE}(u,X)du\right].
$$
Following the derivations in \citet[Eq. (A.29)]{kitagawa2018should},
we obtain for any $\tilde{G}\in\mathcal{G}$ that
\begingroup
\allowdisplaybreaks
\begin{align}
W(\tilde{G})-W(\hat{G}_{hybrid})
=&
\bar{W}(\tilde{G})-\bar{W}(\hat{G}_{hybrid})\notag\\
=&
\bar{W}_n(\tilde{G})-\hat{\bar{W}}_n(\tilde{G})-\bar{W}_n(\hat{G}_{hybrid})+\hat{\bar{W}}_n(\hat{G}_{hybrid})\notag\\&+\bar{W}(\tilde{G})-\bar{W}_n(\tilde{G})+\bar{W}_n(\hat{G}_{hybrid})-\bar{W}(\hat{G}_{hybrid})\notag\\&+\hat{\bar{W}}_n(\tilde{G})-\hat{\bar{W}}_n(\hat{G}_{hybrid})
\notag\\
\leq&
\bar{W}_n(\tilde{G})-\hat{\bar{W}}_n(\tilde{G})-\bar{W}_n(\hat{G}_{hybrid})+\hat{\bar{W}}_n(\hat{G}_{hybrid})\notag\\&+\bar{W}(\tilde{G})-\bar{W}_n(\tilde{G})+\bar{W}_n(\hat{G}_{hybrid})-\bar{W}(\hat{G}_{hybrid})\notag\\
\leq&
E_n\left[(1\{Z\in \hat{G}_{hybrid}\}-1\{Z\in \tilde{G}\})\int_0^1(\widehat{MTE}(u,X)-{MTE}(u,X))du\right]\notag\\&+2\sup_{G\in\mathcal{G}}|\bar{W}_n(G)-\bar{W}(G)|,
\label{eq:eq_a29}
\end{align}
\endgroup
where the first inequality uses $\hat{\bar{W}}_n(\tilde{G})\leq\hat{\bar{W}}_n(\hat{G}_{hybrid})$. 
By Assumption \ref{assn:estimation_error},
we have 
\begin{equation}\label{eq:high_level_assumption}
\sup_{P\in\mathcal{P}_m\cap\mathcal{P}(\bar{M})}E_{P^n}\left[E_n\left[(1\{Z\in \hat{G}_{hybrid}\}-1\{Z\in \tilde{G}\})\int_0^1(\widehat{MTE}(u,X)-{MTE}(u,X))du\right]\right]=O(\psi_n^{-1}).
\end{equation}
By \eqref{eq:lemmaA4} in the proof of Theorem \ref{theorem:known}, 
we have 
\begin{equation}\label{eq:lemmaA4again}
\sup_{P\in\mathcal{P}_m\cap\mathcal{P}(\bar{M})}E_{P^n}\left[\sup_{G\in\mathcal{G}}|\bar{W}_n(G)-\bar{W}(G)|\right]=O(n^{-1/2}).
\end{equation}
The claim of the theorem now follows from \eqref{eq:eq_a29}, \eqref{eq:high_level_assumption}, and \eqref{eq:lemmaA4again}.
\end{proof}

\subsection{Proof of Proposition \ref{theorem:estimation_error}}\label{sec:theorem:estimation_error}
\begin{proof}
From on the structure of $MTE(p,x)$, we have  
\begingroup
\allowdisplaybreaks
\begin{eqnarray*}
&&
E_{P^n}\left[E_n\left[\left|\int_0^1(\widehat{MTE}(u,X)-{MTE}(u,X))du\right|\right]\right]\\
&=&
E_{P^n}\left[E_n\left[\left|X'(\hat\beta_1-\beta_1)-X'(\hat\beta_0-\beta_0)+\sum_{k=2}^K(\hat\alpha_k-\alpha_k)\right|\right]\right]\\
&\leq&
E_{P^n}\left[E_n\left[\|X\|\right](\|\hat\beta_1-\beta_1\|+\|\hat\beta_0-\beta_0\|)\right]+\sum_{k=2}^KE_{P^n}\left[\left|\hat\alpha_k-\alpha_k\right|\right]\\
&\leq&
E_{P^n}\left[(E_n-E)\left[\|X\|\right](\|\hat\beta_1-\beta_1\|+\|\hat\beta_0-\beta_0\|)\right]+E\left[\|X\|\right]E_{P^n}\left[(\|\hat\beta_1-\beta_1\|+\|\hat\beta_0-\beta_0\|)\right]\\&&+\sum_{k=2}^KE_{P^n}\left[\left|\hat\alpha_k-\alpha_k\right|\right].
\end{eqnarray*}
\endgroup
By \eqref{eq:para_compact} and \eqref{eq:X_compact}, it suffices to show 
$$
\limsup_{n\rightarrow\infty}\sup_{P\in\mathcal{P}_m}\psi_nE_{P^n}\left[\left\|\hat\theta-\theta\right\|\right]<\infty.
$$
Since $E_n\left[\hat{\mathcal{X}}(Y-\hat{\mathcal{X}}'\hat\theta)\right]=0$ and $E\left[{\mathcal{X}}(Y-{\mathcal{X}}'\theta)\right]=0$, we can write 
\begin{align*}
E\left[{\mathcal{X}}{\mathcal{X}}'\right] (\hat\theta-\theta)
&=
-(E_n-E)\left[{\mathcal{X}}{\mathcal{X}}'\right]\hat\theta+(E_n-E)\left[{\mathcal{X}}Y\right]\\&-E_n\left[\hat{\mathcal{X}}\hat{\mathcal{X}}'-{\mathcal{X}}{\mathcal{X}}'\right]\hat\theta+E_n\left[(\hat{\mathcal{X}}-{\mathcal{X}})Y\right].
\end{align*}
Therefore, 
\begingroup
\allowdisplaybreaks
\begin{align*}
\left\|\hat\theta-\theta\right\|
&\leq
\lambda_{\min}(E\left[{\mathcal{X}}{\mathcal{X}}'\right])^{-1}
\left\|(E_n-E)\left[{\mathcal{X}}{\mathcal{X}}'\right]\right\|_2\left\|\hat\theta\right\|
\\&
+\lambda_{\min}(E\left[{\mathcal{X}}{\mathcal{X}}'\right])^{-1}\left\|(E_n-E)\left[{\mathcal{X}}Y\right]\right\|
\\&
+\lambda_{\min}(E\left[{\mathcal{X}}{\mathcal{X}}'\right])^{-1}\left\|E_n\left[\hat{\mathcal{X}}\hat{\mathcal{X}}'-{\mathcal{X}}{\mathcal{X}}'\right]\right\|_2\left\|\hat\theta\right\|
\\&
+\lambda_{\min}(E\left[{\mathcal{X}}{\mathcal{X}}'\right])^{-1}\left\|E_n\left[(\hat{\mathcal{X}}-{\mathcal{X}})Y\right]\right\|.
\end{align*}
\endgroup
Therefore, the statement of this theorem follows from \eqref{eq:para_compact}, \eqref{eq:mini_eigen}, Lemma \ref{lemma:uniform_mean}, and Lemma \ref{lemma:hatmathcal{X}_conv}, where the last two lemmas are stated and proved in Appendix \ref{sec:auxiliary_lemmas_for_the_proof_of_proposition}.
\end{proof}

\section{Auxiliary lemmas for the proof of Proposition \ref{theorem:estimation_error}}\label{sec:auxiliary_lemmas_for_the_proof_of_proposition}
\begin{lemma}\label{lemma:uniform_mean}
Under the assumptions of Theorem \ref{theorem:estimation_error}, one has
\begingroup
\allowdisplaybreaks
\begin{align*}
\limsup_{n\rightarrow\infty}\sup_{P\in\mathcal{P}_m}n^{1/2}E_{P^n}\left[(E_n-E)\left[\|X\|\right]\right]<&\infty,
\\
\limsup_{n\rightarrow\infty}\sup_{P\in\mathcal{P}_m}n^{1/2}E_{P^n}\left[\left\|(E_n-E)\left[{\mathcal{X}}Y\right]\right\|\right]<&\infty,
\\
\limsup_{n\rightarrow\infty}\sup_{P\in\mathcal{P}_m}n^{1/2}E_{P^n}\left[\left\|(E_n-E)\left[{\mathcal{X}}{\mathcal{X}}'\right]\right\|_2\right]<&\infty,
\\
\limsup_{n\rightarrow\infty}\sup_{P\in\mathcal{P}_m}n^{1/2}E_{P^n}\left[\left|(E_n-E)\left[\|X\|^2\right]\right|\right]<&\infty,
\qquad\text{and}\\
\limsup_{n\rightarrow\infty}\sup_{P\in\mathcal{P}_m}n^{1/2}E_{P^n}\left[\left|(E_n-E)\left[|Y|^2\right]\right|\right]<&\infty.
\end{align*}
\endgroup
\end{lemma}
\begin{proof}
The statements follow from evaluating the second moment for $E_n-E$ for each random variable. The second moments for these variables are bounded uniformly over $P\in\mathcal{P}_m$ by \eqref{eq:X_compact}.
\end{proof}

\begin{lemma}\label{lemma:hatmathcal{X}_conv}
Under the assumptions of Theorem \ref{theorem:estimation_error}, one has
\begingroup
\allowdisplaybreaks
\begin{align*}
\limsup_{n\rightarrow\infty}\sup_{P\in\mathcal{P}_m}\psi_nE_{P^n}\left[\left\|E_n\left[\hat{\mathcal{X}}\hat{\mathcal{X}}'-{\mathcal{X}}{\mathcal{X}}'\right]\right\|_2\right]<&\infty
\qquad\text{and}\\
\limsup_{n\rightarrow\infty}\sup_{P\in\mathcal{P}_m}\psi_nE_{P^n}\left[\left\|E_n\left[(\hat{\mathcal{X}}-{\mathcal{X}})Y\right]\right\|\right]<&\infty.
\end{align*}
\endgroup
\end{lemma}
\begin{proof} 
Since 
\begingroup
\allowdisplaybreaks
\begin{align*}
{\mathcal{X}}{\mathcal{X}}'=&((1-\nu(Z))X',\nu(Z)X',\nu(Z)^2,\ldots,\nu(Z)^K)'((1-\nu(Z))X',\nu(Z)X',\nu(Z)^2,\ldots,\nu(Z)^K)
\qquad\text{and}\\
\hat{\mathcal{X}}\hat{\mathcal{X}}'=&((1-\hat\nu(Z))X',\hat\nu(Z)X',\hat\nu(Z)^2,\ldots,\hat\nu(Z)^K)'((1-\hat\nu(Z))X',\hat\nu(Z)X',\hat\nu(Z)^2,\ldots,\hat\nu(Z)^K),
\end{align*}
\endgroup
we have  
$$
\|\hat{\mathcal{X}}\hat{\mathcal{X}}'-\mathcal{X}\mathcal{X}'\|_2\leq C_2(\|X\|^2+1)|\hat\nu(Z)-\nu(Z)|
$$ 
for a positive constant $C_2 < \infty$.
Moreover, we have 
\begingroup
\allowdisplaybreaks
\begin{eqnarray*}
&&\left\|(\hat{\mathcal{X}}-{\mathcal{X}})Y\right\|
\\
&=&
\left\|(-(\hat\nu(Z)-\nu(Z))YX',(\hat\nu(Z)-\nu(Z))YX',(\hat\nu(Z)^2-\nu(Z)^2)Y,\ldots,(\hat\nu(Z)^K-\nu(Z)^K)Y)'\right\|\\
&\leq&
|\hat\nu(Z)-\nu(Z)|(2\|X\|+K-1)|Y|.
\end{eqnarray*}
\endgroup
By \eqref{eq:P_hat_estimation} and Lemma \ref{lemma:uniform_mean}, the statement of this lemma holds.
\end{proof}

\end{document}